\documentclass[11pt,reqno,oneside]{amsart}

\usepackage[letterpaper,hmargin=1.3in,vmargin=1.3in]{geometry}

\usepackage{amsmath}
\usepackage{amsthm}
\usepackage{amssymb}
\usepackage{amsfonts}
\usepackage{amsmath}

\usepackage{tikz}
\usetikzlibrary{matrix,arrows,cd,calc}

\usepackage[hidelinks]{hyperref}
\usepackage{xcolor}
\definecolor{darkgreen}{rgb}{0,0.45,0}
\hypersetup{colorlinks,urlcolor=blue,citecolor=darkgreen,linkcolor=darkgreen,linktocpage}

\usepackage[foot]{amsaddr}

\numberwithin{equation}{section}

\usepackage{xspace}
\usepackage{mathtools}
\usepackage[capitalize]{cleveref}

\theoremstyle{plain}
\newtheorem{theorem}{Theorem}[section]

\newtheorem{proposition}[theorem]{Proposition}
\newtheorem{lemma}[theorem]{Lemma}
\newtheorem{corollary}[theorem]{Corollary}

\theoremstyle{definition}
\newtheorem{definition}[theorem]{Definition}

\newtheorem{example}[theorem]{Example}

\newtheorem{question}[theorem]{Question}

\newcommand{\Z}{\mathbb{Z}}
\newcommand{\N}{\mathbb{N}}
\newcommand{\UU}{\mathcal{U}}
\newcommand{\idfunc}[1]{\mathsf{id}_{#1}}
\newcommand{\idb}[1]{\mathsf{idb}_{#1}}
\newcommand{\rec}{\mathsf{rec}}
\newcommand{\End}{\mathsf{end}}
\newcommand{\biinv}{\mathsf{biInv}}
\newcommand{\refl}[1]{\ensuremath{\mathsf{refl}_{#1}}\xspace}
\newcommand{\isBiInv}{\mathsf{isBiInv}}
\newcommand{\prbi}{\mathsf{prBiInv}}

\newcommand{\isSet}{\mathsf{isSet}}
\newcommand{\isProp}{\mathsf{isProp}}
\newcommand{\isContr}{\mathsf{isContr}}
\newcommand{\nf}{\mathsf{nf}}
\newcommand{\ap}[1]{\mathsf{ap}_{#1}}
\newcommand*\sq{\mathbin{\vcenter{\hbox{\rule{.3ex}{.3ex}}}}}
\newcommand{\toEq}{\mathsf{toEq}}
\newcommand{\toBiInv}{\mathsf{toBiInv}}
\newcommand{\El}{\mathsf{El}}
\newcommand{\Map}{\mathsf{Map}}
\newcommand{\bool}{\mathsf{bool}}
\newcommand{\ed}{\mathsf{ed}}
\newcommand{\succc}{\mathsf{succ}}
\newcommand{\pred}{\mathsf{pred}}
\newcommand{\secc}{\mathsf{sec}}
\newcommand{\ret}{\mathsf{ret}}
\newcommand{\coh}{\mathsf{coh}}

\newcommand{\strpos}{\mathsf{strpos}}
\newcommand{\strneg}{\mathsf{strneg}}

\newcommand{\defeq}{\vcentcolon\equiv}  

\title{The Integers as a Higher Inductive Type}
\author{Thorsten Altenkirch}
\author{Luis Scoccola}

\date{}

\begin{document}

\begin{abstract}
  We consider the problem of defining the integers in Homotopy Type
  Theory (HoTT). We can define the type of integers as signed natural numbers
  (i.e., using a coproduct), but its induction principle is very inconvenient to work with,
  since it leads to an explosion of cases. An alternative is to use
  set-quotients, but here we need to use set-truncation to avoid
  non-trivial higher equalities. This results in a recursion principle
  that only allows us to define function into sets (types satisfying UIP).
  In this paper we consider higher
  inductive types using either a small universe or bi-invertible maps.
  These types represent integers without explicit set-truncation that
  are equivalent to the usual coproduct representation.
  This is an interesting example since it shows how some coherence
  problems can be handled in HoTT. We discuss some open questions
  triggered by this work.
  The proofs have been formally verified using cubical Agda.
\end{abstract}

\maketitle


\section{Introduction}
\label{sec:introduction}

How to define the integers in Homotopy Type Theory? This can sound like a trivial
question. The first answer is as signed natural numbers:
\begin{definition}\label{Z-w}
Let $\Z_w$ be the inductive type generated by the following constructors:
  \begin{itemize}
  \item[--] $0 : \Z_w$
  \item[--] $\strpos : \N \to \Z_w$
  \item[--] $\strneg : \N \to \Z_w$
  \end{itemize}
\end{definition}
However, this type is very inconvenient in practice because it creates
a lot of unnecessary case distinctions. Nuo \cite{nuo-phd} tried to
prove distributivity of multiplication over addition, which
resulted in a lot of cases. It is like working with normal forms only,
when working with $\lambda$-terms.

Nuo shows that it is much better to work with a quotient type, 
representing integers as differences of natural
numbers. That is, we define $\Z_q = \N \times \N / \sim$ where 
$(x^+,x^-) \sim (y^+,y^-)$ is defined as $x^+ + y^- = y^+ + x^-$
\footnote{This is actually the definition in \cite{hottbook}.}.
However, this is not the end of the story. Here we use
set-quotients, which can be implemented as a higher inductive type with a
set-truncation constructor \cite[Section~6.10]{hottbook}.
However, the set-truncation constructor implies that using its recursion principle we can only define
functions into sets, which seems to be an unreasonable limitation when
working in HoTT. For example, in the proof that the loop space of the
circle is isomorphic to the integers \cite{licatashulman}, we must map from the integers to the loop space
of the circle, when we do not yet know that this will end up being a set.

We would like to have a definition of the integers which is convenient
to work with (i.e., does not reduce them to normal forms) but which is
not forced to be set-truncated by a set-truncation constructor.
Paolo Capriotti suggested the following definition:
\begin{definition}\label{Z-h}
Let $\Z_h$ be the higher inductive type with the following constructors:
\begin{itemize}
    \item[--] $0 : \Z_h$;
    \item[--] $\succc : \Z_h \to \Z_h$;
    \item[--] $\pred : \Z_h \to \Z_h$;
    \item[--] $\secc : (z : \Z_h) \to \pred(\succc(z)) = z$;
    \item[--] $\ret : (z : \Z_h) \to \succc(\pred(z)) = z$;
    \item[--] $\coh : (z : \Z_h) \to \ap{\succc}(\secc(z)) = \ret(\succc(z))$.
\end{itemize}
\end{definition}
We add $\succc$ and $\pred$ as constructors, but then we postulate that they
are inverse to each other using $\secc$ and $\ret$. At this point we could add
a set-truncation but then we would suffer from the same shortcoming as
the definition using a set-quotient. However, we can add just one
coherence condition $\coh$ which should look familiar to anybody
who has read the HoTT book: indeed the constructors $\pred$, $\secc$, $\ret$, and $\coh$
exactly say that $\succc$ is a half-adjoint equivalence \cite[Section~4.2]{hottbook}. 
More precisely, $\secc$ postulates that $\succc$ is a section,
$\ret$ postulates that $\succc$ is a retraction, and $\coh$ represents the
triangle identity in the definition of half-adjoint equivalence.

The question that now remains is the following. Is $\Z_h$ a correct definition of the
integers, in particular is it a set with decidable equality? The
strategy to prove this is to define a normalisation function into the
signed integers, $\Z_w$, and show that this normalisation function,
together with the obvious embedding of $\Z_w$ into $\Z_h$, forms an equivalence.
It turns out that this is actually quite hard to prove, due to the presence of
higher equalities, and nobody has so far been able to formally verify this.

In this paper, we follow the same idea but use a \emph{simpler}
definition of equivalence, namely bi-invertible maps \cite[Section~4.3]{hottbook}:
\begin{definition}\label{Z-b}
Let $\Z_b$ be the higher inductive type with the following constructors:
\begin{itemize}
    \item[--] $0 : \Z_b$;
    \item[--] $\succc : \Z_b \to \Z_b$;
    \item[--] $\pred_1 : \Z_b \to \Z_b$;
    \item[--] $\pred_2 : \Z_b \to \Z_b$;
    \item[--] $\secc : (z : \Z_b) \to \pred_1(\succc(z)) = z$;
    \item[--] $\ret : (z : \Z_b) \to \succc(\pred_2(z)) = z$;
\end{itemize}
\end{definition}
In this case we postulate that $\succc$ has a left inverse, given by
$\pred_1$ and $\secc$, and a right inverse, given by $\pred_2$ and $\ret$.
The reason why $\Z_b$ is simpler than $\Z_h$ is because it only has
$0$- and $1$-dimensional constructors. The higher coherence $\coh$ is not needed
in this case for the same reason that a $2$-dimensional constructor is not needed
in the definition of bi-invertible map: having two, a priory, unrelated
inverses makes the type of witnesses that a certain map is bi-invertible a proposition
(\cite[Theorem~4.3.2]{hottbook}).

For this definition we can give a complete
proof that $\Z_b$ is equivalent to $\Z_w$, which has been formalized in
cubical Agda. We remark that this has previously been verified by Evan Cavallo \cite{cavallo} in
RedTT \cite{redtt}. However, our approach to prove the equivalence is more general.
Our main result is \cref{enough}, which says that only the components witnessing the
preservation of $0$ and $\succc$ are relevant when comparing morphisms out of $\Z_b$.

Another presentation of the integers follows from directly
implementing the idea that the integers can be specified as the initial type
with an inhabitant and an equivalence:
\newpage
\begin{itemize}
    \item[--] $0 : \Z_U$;
    \item[--] $s : \Z_U = \Z_U$.
\end{itemize}
The problem is that this is not a standard definition of a higher inductive type
because we state an equality of the type itself. However, this can be
fixed by using a small universe: 

\begin{definition}
\label{Z-u}
  Define $U:\UU$ and $\El : U \to \UU$ inductively with the
  constructors:
  \begin{itemize}
  \item[--] $z : U$;
  \item[--] $q : z = z$;
  \item[--] $0 : \El(z)$
  \end{itemize}
  Now, let $\Z_U \defeq \El(z)$.
\end{definition}
While we can show that this is
a set without using set-truncation, its recursion principle isn't directly amenable to recursive
definitions of functions because even $\succc$ is not a constructor. On
the other hand the fact that the integers are the loop space of the
circle is a rather easy consequence of this definition. 

The definition of the integers is also closely related to the free
group, indeed as suggested in \cite{free-group} we can define the free
group over a type $A$ by simply parametrizing all the
constructors but $0$:
\begin{definition}
\label{FG}
  Given $A:\UU$, define $\mathbf{F}(A)$ inductively with the
  constructors:
\begin{itemize}
    \item[--] $0 : \mathbf{F}(A)$;
    \item[--] $\succc : A \to \mathbf{F}(A) \to \mathbf{F}(A)$;
    \item[--] $\pred_1 : A \to \mathbf{F}(A) \to \mathbf{F}(A)$;
    \item[--] $\pred_2 : A \to \mathbf{F}(A) \to \mathbf{F}(A)$;
    \item[--] $\secc : (a : A) \to (z : \mathbf{F}(A)) \to \pred_1(a,\succc(a,z)) = z$;
    \item[--] $\ret : (a : A) \to (z : \mathbf{F}(A))\to \succc(a,\pred_2(a,z)) = z$;
\end{itemize}
\end{definition}
The integers arise as the special case $\Z =
\mathbf{F}(\mathbf{1})$. However, the normal forms get a bit more
complicated because we must allow alternating sequences of $\succc$ and
$\pred$ but only for different $a:A$. This means that a normalisation
function is only definable for sets $a:A$ with a decidable
equality. The general problem of whether $\mathbf{F}(A)$ is a set, if $A$ is, is
still open --- in \cite{free-group} it is shown to be the
case, if we $1$-truncate the HIT. 

The problem of defining the integers with convenient constructors, and adding only
the right coherences to make it a set, can be seen as a simple instance of a
more general class of coherence problems in HoTT. Another example that we have in mind is the intrinsic definition of the syntax of
type theory as the initial category with families as developed in
\cite{TTinTT}. If we carry out this definition in HoTT, we need to
set-truncate the syntax, but this stops us from interpreting the
syntax in the standard model formed by sets. We hope that also in this case we can
add the correct coherence laws and show that they are sufficient to deduce
that the initial algebra is a set.

\subsection{Contributions}
\label{sec:contributions}

We show that the definitions of the signed integers, $\Z_w$, the
definition of the integers as a higher inductive type using bi-invertible maps, $\Z_b$,
and the definition using a higher inductive-inductive type with a mini universe, $\Z_U$, are all
equivalent (\cref{Zb-equiv-Zw}).

For $\Z_b$ we establish some
useful principles such as a recursion principle (\cref{mapout})
which only uses one predecessor, and an induction principle which says
that to prove a predicate (i.e., a family of propositions), you
only need to prove closure under $0$, $\succc$, and $\pred_1$
(\cref{proveprop}). This is sufficient to verify all algebraic
properties of the integers, e.g., that the integers form a commutative ring. We
have formalized \cite{code} the constructions using cubical Agda \cite{cubical-agda}.

When formalizing the constructions involving $\Z_b$ we developed the
theory of bi-invertible
maps in cubical Agda, which wasn't available.
In particular, we prove that bi-inverti\-ble maps are equivalent to contractible-fibers maps
\cite[Section~4.4]{hottbook}, and the principle of equivalence induction for bi-invertible maps.

\subsection{Related work}
\label{sec:related-work}

The claim that $\Z_h$ is a set can be found in \cite{pinyo-types} but
the proof was flawed: it relies on the assumption that we can ignore
propositional parts of an algebra for a certain signature when constructing algebra morphisms,
which is not the case in general (\cref{counterexample}). Cavallo \cite{cavallo} verified that
$\Z_b \simeq \Z_w$ in RedTT. Higher inductive representations of
the integers are discussed in \cite{BasoldEtAl} and it is shown there that
$\Z_h$ without the last constructor is not a
set. \cite{nicolai-path-19} also discuss \cite{pinyo-types} and note
that it is a corollary of their higher Seifert-van Kampen
theorem --- however, they derive it from initiality not from the
induction principle.

\subsection{Background}
\label{background}

We use Homotopy Type Theory as presented in the book
\cite{hottbook}. We adopt the following notational conventions.

If two terms $a$ and $b$ are definitionally equal, we write $a \equiv b$,
and we reserve $a = b$ to denote the type of propositional equalities between $a$ and $b$.

Given a type $A : \UU$ and a type family $P : A \to \UU$, we write the corresponding
$\Pi$-type as $(a : A) \to P(a)$, and the corresponding $\Sigma$-type as
$(a : A) \times P(a)$.

Given a type $A : \UU$, a type family $P : A \to \UU$, an equality $e : a = b$ in $A$,
and $p : P(a)$, we denote the coercion of $p$ along $e$ by $e_*(p) : P(b)$.
This is defined by induction on $e$.

A type is contractible if it has exactly one inhabitant. That is,
given a type $A : \UU$, we define $\isContr(A) \defeq (a_0 : A) \times \left((a : A) \to a = a_0\right)$.
A type is a proposition if any two inhabitants are equal. That is,
given a type $A : \UU$, we define $\isProp(A) \defeq (a,b : A) \to a = b$, \cite[Definition~3.3.1]{hottbook}.
A type is a set if it satisfies UIP. That is, given a type $A : \UU$,
we define $\isSet(A) \defeq (a,b : A) \to (p, q : a = b) \to p = q$, \cite[Definition~3.1.1]{hottbook}.

An equivalence between types $A$ and $B$ is a map $f : A \to B$ together
with a proof that $(b : B) \to \isContr( (a : A) \times f(a) = b)$.
The type of equivalences between $A$ and $B$ is denoted by $A \simeq B$.

The general syntax of Higher Inductive Inductive Types (HIITs) is
specified in \cite{ambrus-andras}, where also the types of the
eliminators are derived.
In the informal exposition, and in the formalisation, we use
the cubical approach to path algebra introduced in \cite{licata-brunerie}.

For the formalisation we use cubical Agda \cite{cubical-agda} which is
based on the cubical type theory of \cite{cubical}. The development of HIITs
in Agda is based on \cite{cubical-hits}.

\section{Representing $\Z$ using bi-invertible maps}
\label{sec:usingbi}
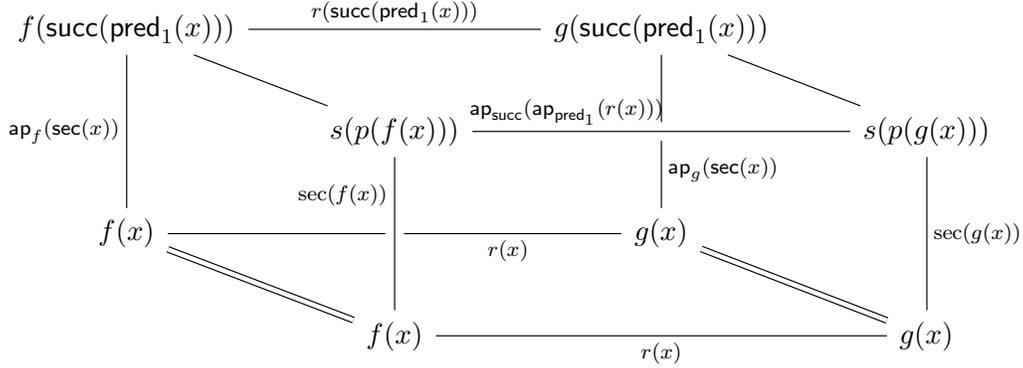
\begin{figure*}[!ht]
\[\begin{tikzcd}
 f(\succc(\pred_1(x))) \arrow[dd,dash,"\ap{f}(\secc(x))" left]
 \arrow[rr,dash,"r(\succc(\pred_1(x)))"]  
 \arrow[dr,dash]
 & & g(\succc(\pred_1(x))) \arrow[dd,near end,dash,"\ap{g}(\secc(x))"]
 \arrow[dr,dash]\\
 & s(p(f(x))) \arrow[rr,near start,dash,crossing over,"\ap{\succc}(\ap{\pred_1}(r(x)))"]  
& & s(p(g(x))) \arrow[dd,dash,"\sec(g(x))" right]\\
 f(x) \arrow[rr, near end,dash, "r(x)" below] 
 \arrow[dr,equal]
&& g(x) \arrow[dr,equal]\\
 & f(x) \arrow[rr,dash,"r(x)" below]  
\arrow[from=uu,dash,near start,crossing over, "\sec(f(x))" left]
&& g(x)
\end{tikzcd}
\]  
\caption{\label{cube}Cube needed for lemma \ref{uniquenessZb}}
\end{figure*}
The type $\N$ of natural numbers is usually defined as the inductive type generated by an inhabitant $0 : \N$ and
an endomap $\succc : \N \to \N$. In this section, we define the integers $\Z_b$ in a similar way. The idea is to give
constructors that guarantee that we have $0 : \Z_b$, $\succc : \Z_b
\to \Z_b$, and that $\succc$ is an equivalence using bi-invertible
maps, see \cref{Z-b}.
To make it easy to work with this definition, we prove three theorems that let us: map out of $\Z_b$ (\cref{mapout}), prove properties
about $\Z_b$ (\cref{proveprop}), and recognise when two maps out of $\Z_b$ are equal (\cref{enough}).



The result about mapping out of $\Z_b$ is very simple, and follows immediately from the recursion principle of $\Z_b$.

\begin{proposition}[{\texttt{rec}$\Z$\texttt{simp}}]
    \label{mapout}
Given a type $T$ with an inhabitant $t : T$ and two maps $f : T \to T$, $g : T \to T$, such that
$g$ is a left and right inverse of $f$, we get a map $r : \Z_b \to T$ such that
$r(0) \equiv t$ and $r(\succc(z)) \equiv f(r(z))$, definitionally.
\end{proposition}

The next result is only slightly more involved.

\begin{proposition}[{\texttt{ind}$\Z$\texttt{simp}}]
    \label{proveprop}
    Given a type family $P : \Z_b \to \UU$ such that $(z : \Z_b) \to \isProp(P(z))$,
    if we have $P(0)$, $(z : \Z_b) \to P(z) \to P(\succc(z))$, and $(z : \Z_b) \to P(z) \to P(\pred_1(z))$,
    then it follows that $(z : \Z_b) \to P(z)$.
\end{proposition}
\begin{proof}
    We use the induction principle of $\Z_b$.
    The main idea is that we do not have to check any coherences, since
    we are proving a proposition. Concretely, this means that we only have to provide inhabitants
    for the following types:
    $P(0)$, $(z : \Z_b) \to P(z) \to P(\succc(z))$, $(z : \Z_b) \to P(z) \to P(\pred_1(z))$,
    and $(z : \Z_b) \to P(z) \to P(\pred_2(z))$.
    For the first three we just use the assumptions.
    For the fourth one, we make use of the fact that, for every $z : \Z_b$, there
    is an equality $\pred_1(z) = \pred_2(z)$.
    This is because $\pred_2(z) = \pred_1(\succc(\pred_2(z))) = \pred_1(z)$ using
    $\secc$ and then $\ret$.
\end{proof}


The result that allows us to compare maps out of $\Z_b$ is considerably more complicated to prove.
In order to explain its proof, we need to talk about bi-invertible maps.

\begin{definition}
    A map between types $f : A \to B$ is a bi-invertible map if
    there exist $g,h : B \to A$, and homotopies $s : g\circ f = \idfunc{A}$ and $r : f\circ h = \idfunc{B}$.

    The type of bi-invertible structures on such a map $f$ is denoted by $\isBiInv(f)$.
    The type $(f : A \to B) \times \isBiInv(f)$ is denoted by $A \simeq_b B$.
\end{definition}

Whenever we have $f : A \simeq_b B$, we will abuse notation, and write $f : A \to B$ for the underlying
function of the bi-invertible map $f$.

Notice that the constructors $\succc$, $\pred_1$, $\pred_2$, $\secc$, and $\ret$ form a bi-invertible map.
Suppose given a type $T$ with an inhabitant $t : T$ and a bi-invertible map $s : T \simeq_b T$.
The recursion principle of $\Z_b$ gives us $\rec_{\Z_b}(T,t,s) : \Z_b \to T$.
Now, assume given another map $f : \Z_b \to T$.
What do we have to check to be able to conclude that $f = \rec_{\Z_b}(T,t,s)$?

The following theorem gives a simple answer to the question and is the main
focus of this section.

\begin{theorem}[\texttt{uniqueness}$\Z$]
    \label{enough}
    Given a type $T$, an inhabitant $t : T$, a bi-invertible map $s : T \simeq_b T$, and a map $f : \Z_b \to T$,
    if $f(0) = t$ and $s \circ f = f \circ \succc$
    then $f = \rec_{\Z_b}(T,t,s)$.
\end{theorem}

In order to prove \cref{enough} we must study the preservation of bi-invertible maps,
which we introduce next.

Fix types $A,B,A',B' : \UU$, bi-invertible maps $e : A \simeq_b B$ and $e' : A' \simeq_b B'$, and
maps $\alpha : A \to A'$ and $\beta : B \to B'$:
\[
\begin{tikzpicture}
  \matrix (m) [matrix of math nodes,row sep=2em,column sep=2em,minimum width=2em,nodes={text height=1.75ex,text depth=0.25ex}]
    { A & B \\
    A' & B'. \\};
  \path[-stealth]
    (m-1-1) edge node [above] {$e$} (m-1-2)
            edge node [left] {$\alpha$} (m-2-1)
    (m-2-1) edge node [above] {$e'$} (m-2-2)
    (m-1-2) edge node [right]{$\beta$} (m-2-2)
    ;
\end{tikzpicture}
\]
We now define what it means for $\alpha$ and $\beta$ to respect $e$ and $e'$.
By a slight abuse of notation, let the bi-invertible maps $e$ and $e'$ be given
by $(e,g,h,s,r)$ and $(e',g',h',s',r')$.

\begin{definition}
    We define the type $\prbi(e,e',\alpha,\beta)$ as the iterated $\Sigma$-type with the following fields:
    \begin{itemize}
        \item[--] (preservation of $e$) $p_e : e' \circ \alpha = \beta \circ e$;
        \item[--] (preservation of $g$) $p_g : g' \circ \beta = \alpha \circ g$;
        \item[--] (preservation of $h$) $p_h : h' \circ \beta = \alpha \circ h$;
        \item[--] (preservation of $s$) $p_s : (a : A) \to s'( \alpha(a)) = \ap{g'}(p_e a) \sq p_g(e(a)) \sq \ap{\alpha}(s(a))$;
        \item[--] (preservation of $r$) $p_r : (b : B) \to r'( \beta(b)) = \ap{e'}(p_h b) \sq p_e(h(b)) \sq \ap{\beta}(r(a))$.
    \end{itemize}
\end{definition}

The next proposition follows from the initiality of $\Z_b$, although it is a bit involved to prove formally using
the constructors and the induction principle.

\begin{proposition}[\texttt{uniqueness}]
    \label{uniquenessZb}
    Suppose given a type $T$ with an inhabitant $t : T$, a bi-invertible map $s : T \simeq_b T$, and a map $f : \Z_b \to T$.
    If $f(0) = t$ and $\prbi(\succc,s,f,f)$, then $f = \rec_{\Z_b}(T,t,s)$.
\end{proposition}
\begin{proof}
We write $g$ for $\rec_{\Z_b}(T,t,s)$.
By function extensionality, it is enough to construct a term $r :
\Pi_{x : \Z_b} f(x) = g(x)$. We do this using the induction principle.
The case for $0$
follows directly from the assumption $f(0) = t$, and $r(\succc(x)) = \ap{s}(g(x))$
and the corresponding equalities for $\pred_1$ and $\pred_2$
follow directly from the assumption that $f$ respect the bi-invertible
maps $\succc$ and $s$.

It remains to check the cases of $\secc$ and $\ret$.
Since these are symmetric, we only describe the case of $\secc$.
In this case, we have to provide a filler for the following
square of equalities:
\[\begin{tikzcd}[column sep=huge]
 f(\succc(\pred_1(x))) \arrow[d,dash,"\ap{f}(\secc(x))" left]
 \arrow[r,dash,"r(\succc(\pred_1(x)))"]  
 & g(\succc(\pred_1(x))) \arrow[d,dash,"\ap{g}(\secc(x))"]\\
 f(x) \arrow[r,dash,"r(x)" below] & g(x).
\end{tikzcd}\]  
This filler can be obtained by filling the cube in figure \ref{cube}, as follows.
All the sides apart from the square in question can be filled
using the fact that $f$ preserves the bi-invertible maps, and simple
path algebra, so we can conclude the proof
using the Kan filling property of cubes: any open box can be filled.
\end{proof}

Given a type $A : \UU$, let $\idfunc{A} : A \to A$ be the identity function. We have $\idb{A} : \isBiInv(\idfunc{A})$,
so we can define a map $\toBiInv : A = B \to A\simeq_b B$ by path induction, sending
$\refl{} : A = A$ to $\idb{A}$.
By \cite[Corollary~4.3.3]{hottbook} and the univalence axiom, the map $\toBiInv$ is an equivalence.
Let $\toEq : A \simeq_b B \to A = B$ be its inverse.

From this we can derive the principle of (based) equivalence induction,
which we now state.

\begin{lemma}[\texttt{BiInduction}]\label{BiInduction}
Fix a type $A : \UU$ and a type family $P : (B : \UU) \to A \simeq_b B \to \UU$.
If we have $P_0 : P(A,\idfunc{A},\idb{A})$, then we have
    $(B : \UU) \to (e : A \simeq_b B) \to P(B,e)$.
\end{lemma}
\begin{proof}
This is proven by path induction, after translating bi-invertible maps to equalities, using $\toEq$ and $\toBiInv$.
\end{proof}

Using equivalence induction, and singleton elimination, one can finally prove that a map between types together with bi-invertible maps
that respects the maps, automatically respects the bi-invertible structure.

\begin{lemma}
    \label{inductiontoalgebra}
    The type $\prbi(e,e',\alpha,\beta)$ is equivalent to the type $e' \circ \alpha = \beta \circ e$.
\end{lemma}
\begin{proof}
We use equivalence induction (\cref{BiInduction}) for $e$ and $e'$ and then observe that
the type
\[
    \prbi((\idfunc{},\idb{}),(\idfunc{},\idb{}),\alpha,\beta)
\]
is equivalent to the type of equalities $\alpha = \beta$.
\end{proof}

\begin{proof}[Proof of \cref{enough}]
    The theorem is a corollary of \cref{uniquenessZb} and \cref{inductiontoalgebra}.
\end{proof}

One should notice that \cref{inductiontoalgebra} can be proven directly, avoiding the usage of
the univalence axiom (which was used to prove that $\toBiInv$ is an equivalence). The reason
why we don't do this, is because the path algebra involved in proving \cref{inductiontoalgebra}
directly is non-trivial.

\section{$\Z$ is a set}
\label{sec:isset}

In this section we relate $\Z_b$ with the usual definition of the integers as signed natural numbers,
which we call $\Z_w$.
We show that $\Z_b \simeq \Z_w$, and since we already know that $\Z_w$ is a set,
we deduce that $\Z_b$ is a set too.

%

\begin{definition}
Let $\Z_w$ be the inductive type with the following constructors:
\begin{itemize}
    \item[--] $0 : \Z_w$
    \item[--] $\strpos : \N \to \Z_w$
    \item[--] $\strneg : \N \to \Z_w$
\end{itemize}
\end{definition}

\begin{theorem}[$\Z$\texttt{is}$\Z$]
    We have an equivalence $\Z_b \simeq \Z_w$.
\end{theorem}

\begin{proof}
On the one hand, one can define $\succc_w : \Z_w \to \Z_w $ by induction, by mapping:
\begin{itemize}
    \item[--] $0 \mapsto \strpos(0)$;
    \item[--] $\strpos(n) \mapsto \strpos(\succc(n))$;
    \item[--] $\strneg(0) \mapsto 0$;
    \item[--] $\strneg(\succc(n)) \mapsto \strneg(n)$.
\end{itemize}
Similarly one defines $\pred_w$.
The fact that $\pred_w$ provides a left and right inverse for $\succc_w$ is straightforward.
So, by \cref{mapout} we get a map $\nf : \Z_b \to \Z_w$.
On the other hand, it is easy to construct a map $i : \Z_w \to \Z_b$ by induction.

Induction on $\Z_w$ shows that $\nf \circ i = \idfunc{\Z_w}$.
The hard part is to show that $i \circ \nf = \idfunc{\Z_b}$.
This is where \cref{enough} comes in handy. \cref{enough} implies that it is enough to check that
$(i \circ \nf)(0) = 0$ and that $\succc \circ (i \circ \nf) = (i \circ \nf) \circ \succc$,
and this follows directly by construction.
\end{proof}

\section{Representing $\Z$ using a universe}
\label{sec:usinguniverse}

In this section we give another definition of the integers, denoted by $\Z_U$, which allows one to easily prove that they are the initial
type together with an inhabitant and an equality from the type to itself.

%
To make sense of initiality, we first define the type of $\Z$-algebras and of $\Z$-algebra morphisms.

\begin{definition}
    A $\Z$-algebra is a type $T : \UU$ together with an inhabitant $t : T$,
    and an equality $e : T = T$.
    We denote such a $\Z$-algebra as $(T,t,e)$, or $T$ if the rest of the structure
    can be inferred from the context.
\end{definition}

\begin{definition}
    A morphism of $\Z$-algebras from $(T,t,e)$ to $(T',t',e')$
    is given by a map $f : T \to T'$, together with an equality $f(t) = t'$, and a proof that
    $e_*(f) = e'_*(f)$.
    We denote the type of morphisms of $\Z$-algebras between $T$ and $T'$ by $T \to_{\Z} T'$.
\end{definition}

We are interested in initial $\Z$-algebras.

\begin{definition}
    A initial $\Z$-algebra is a $\Z$-algebra $(T,t,e)$ such that
    for any other $\Z$-algebra $(T',t',e')$ the type
    $T\to_{\Z} T'$ is contractible.
\end{definition}




See \cref{Z-u} for the definition of the initial
$\Z$-algebra using a mini universe%
\footnote{This is inspired by Zongpu Xie's proposal how to represent
  HIITs in Agda \cite{ambrus-agda}.}.
Then define an interpretation function $\El : U \to \UU$, as the higher inductive family with
only one constructor $0 : \El(z)$.
Define the type $\Z_U \defeq \El(z)$. The type $\Z_U$ has the structure of a $\Z$-algebra, since we have
$0 : \Z_U$ and $s \defeq \ap{\El}(q) : \Z_U = \Z_U$.
The following result follows by a routine application of the induction
principle of $\Z_U$.

\begin{theorem}[$\Z$\texttt{isInitial}]
    The $\Z$-algebra $\Z_U$ is initial.\qed
\end{theorem}

In particular, we have.

\begin{proposition}
Given a type $T$ with an inhabitant $t : T$ and an equality $e : T = T$, we get a morphism of $\Z$-algebras $\Z_U \to T$.\qed
\end{proposition}

Again, comparing maps out of $\Z_U$ is easy, thanks to the following theorem.

\begin{theorem}
    Given a type $T$, an inhabitant $t : T$, an equality $e : T = T$, and a map $f : \Z_U \to T$,
    if $f(0) = t$ and $e_* \circ f = f \circ s_*$ then $f = \rec_{\Z_U}(T,t,e)$.\qed
\end{theorem}
Analogously to the case of $\Z_b$, this is proven by combining the initiality of $\Z_U$ with
the fact that to preserve an equality in the universe $e : T = T$, it is enough to commute with
its corresponding coercion function $e_* : T \to T$.

Following the argument given in \cref{sec:isset}, one deduces the following.

\begin{theorem}
    \label{Zb-equiv-Zw}
  There is an equivalence $\Z_U \simeq \Z_w$.\qed
\end{theorem}

We omit the proof since it is basically the same as the construction
presented in \cite{licatashulman} when proving that the integers are
the loop space of the circle.

Indeed, the mini universe $U$ is nothing but the higher inductive type
presentation of the circle $S^1$ of \cite[Section~6.1]{hottbook}, so
that $(z = z) \equiv \Omega S^1$.
Moreover, the type family $\El$ is equivalent to the path space fibration of the circle, in the following sense.

\begin{theorem}[\texttt{ElisPath}]
    For every $u : U$ we have $\ed(u) : \El(u) \simeq (z = u)$.
\end{theorem}
\begin{proof}
    We construct a map $\ed(u) : \El(u) \to (z=u)$
    using induction on $U$ and mapping $0 : \El(z)$ to $\refl{z}$,
    To construct a map going the other way,
    we use path induction and map $\refl{z}$ to $0$.
    It is then straightforward to see that these maps
    give an equivalence as in the statement.
\end{proof}

As a corollary, we obtain the well-known equivalence between the loop space of the circle and the integers.

\begin{corollary}
    We have an equivalence $\Omega S^1 \simeq \Z_w$.\qed
\end{corollary}

This suggests that alternatively one could view the representation of the integers as a
universe as an inductive-inductive presentation of the circle equipped
with a family that has a point in the fiber over the base point. 




\section{Formalization in cubical Agda}

We formally checked the results of this paper \cite{code} using cubical Agda \cite{cubical-agda}.
There are two differences between the informal presentation in the paper and the formalisation.
The first one is that the presentation in the paper is done using book-HoTT \cite{hottbook}, whereas
the formalisation is done using a cubical type theory. In this case, this difference is not important, since
it is easy to translate the formalized arguments to book-HoTT.

The real difference is in the definition of higher inductive types. In the paper we define higher inductive types
as initial algebras for a certain signature (\cref{background}). In the formalisation, we use higher inductive types
as implemented in cubical Agda, which are based on \cite{cubical-hits}.
Although it is natural to assume that the Agda higher inductive type should be initial in the sense of \cref{background},
proving this fact is actually one of the main difficulties in the formalisation (\cref{uniquenessZb}).

In proving the results of \cref{sec:usingbi}, we developed the theory of bi-invertible maps in cubical Agda, which
wasn't available. We prove that the type of bi-invertible maps between $A$ and $B$ is equivalent to the type of equivalences
between $A$ and $B$, and the principle of bi-invertible induction.

\section{Open questions}
\label{sec:concl-open-quest}

\subsection*{Preservation of properties}

The key result in the above discussion is \cref{inductiontoalgebra}, which can be reformulated as follows.
Let $T,T' : \UU$, $s : T \to T$, $s' : T' \to T'$, $\phi : \isBiInv(s)$, and $\phi' : \isBiInv(s')$.
We can define the following two types of morphisms between $T$ and $T'$:
\begin{align*}
    \Map_{\End}(T,T') &\defeq (f : T \to T') \times (s' \circ f = f \circ s)\\
    \Map_{\biinv}(T,T') &\defeq (f : T \to T') \times \prbi(s,s',f,f).
\end{align*}
Informally, $\Map_{\End}(T,T')$ is the type of maps that respect the endomorphism, and $\Map_{\biinv}(T,T')$ is
the type of maps that respect the endomorphism and the proof that the endomorphism is a bi-invertible map.

We have a forgetful map $\Map_{\biinv}(T,T') \to \Map_{\End}(T,T')$, and what \cref{inductiontoalgebra} says is that
this map is an equivalence.

There is something special about the type family $\isBiInv : (A \to B) \to \UU$, and that is that it is valued
in propositions.
One might wonder if \cref{inductiontoalgebra} is a general principle, in the following sense.
Say that we have a signature $S$ for a type of algebras, and we extend it to a signature $S'$, such that
the fields we added take values in propositions. In the above example $S$ corresponds to $(T : \UU) (f : T \to T)$
and $S'$ corresponds to the extension $(T : \UU) (f : T \to T) (\phi : \isBiInv(s))$.
As usual, given $S'$-algebras $T,T'$, we have a forgetful map $\Map_{S'}(T,T') \to \Map_S(T,T')$.
Is this map an equivalence in general?

The following example, suggested by Paolo Capriotti, shows that this is not necessarily the case.
\begin{example}
    \label{counterexample}
    Consider $S$, the signature $(T : \UU) (o : T \to T \to T)$, and $S'$, the extension
    \begin{align*}
        &(T : \UU) (o : T \to T \to T) (tr : \isSet(T))\\
        &(e : T) \left(u : (t : T) \to o(t,e) = t \times o(e,t) = t\right).
    \end{align*}
    The $S$-algebras are the types with a binary operation, and the $S'$-algebras are the sets with 
    a binary operation with a distinguished element that is a left and right unit.

    The extension $S'$ is propositional. This is because being a set is a proposition (\cite[Theorem~7.1.10]{hottbook}),
    so $tr$ inhabits a proposition, two left and right units must necessarily coincide,
    so $e$ inhabits a proposition (assuming $u$),
    and the identity types of a set are propositions and these are closed under pi-types (\cite[Theorem~7.1.9]{hottbook}),
    so $u$ inhabits a proposition (assuming $tr$).

    Let us see that for $S'$-algebras $T,T'$ the forgetful map $\Map_{S'}(T,T') \to \Map_S(T,T')$ is not an
    equivalence in general. Let $T$ be $(\N,+,\phi,0,\psi)$, where $\phi$ is a proof that the natural numbers form
    a set, and $\psi$ is a proof that $0$ is a left and right unit for $+$.
    Let $T'$ be $(\bool,\vee,\phi',\bot,\psi')$, where $\phi'$ is a proof that the booleans form a set,
    and $\phi$ is a proof that $\bot$ is a left and right unit for $\vee$.
    Then we have $\lambda n. \top : \N \to \bool$. This map clearly respects the operations, so we get an inhabitant
    of $\Map_S(T,T')$. But this morphism does not respect the units, so it cannot come from a morphism
    in $\Map_{S'}(T,T')$.
\end{example}


This discussion leaves open an interesting question.
\begin{question}
    Given a signature $S$ and a propositional extension $S'$, are there useful necessary and sufficient conditions for
    the forgetful map $\Map_{S'}(T,T') \to \Map_S(T,T')$ to be an equivalence for every pair of $S'$-algebras
    $T$ and $T'$?
\end{question}

\subsection*{Initiality of HIITs}
\label{sec:intiality-hiits}

Our original goal was to complete the conjectured result from
\cite{pinyo-types} and formally verify that $\Z_h$ is a set. Using the
strategy from this paper this is fairly straightforward: we can show
that the natural notion of morphism of $\Z_h$-algebras satisfies a
principle analogous to \cref{inductiontoalgebra}, and hence that
$\Z_h$ is a set. When attempting to formalize this
construction we hit an unexpected problem: it turns out that it is rather
difficult to verify that the higher inductive type defining $\Z_h$ is initial in its corresponding
wild category of algebras. Specifically, the proof seems to require the construction of a
filler for a $4$-dimensional cube which is rather laborious. In \cite{qiits-popl} it
is shown that for QIITs (i.e., set-truncated HIITs) elimination and
initiality are equivalent, but the extension to higher dimensional HIITs
seems non-trivial. In particular it may require developing the higher
order categorical structure of the category of algebras.

\paragraph{Acknowledgements}
  The first author would like to thank Paolo Capriotti, Nicolai Kraus
  and Gun Pinyo for many interesting discussions on the subject of
  this paper. Both authors would like to thank Christian Sattler for
  comments and useful discussions. The work by and with Ambrus Kaposi
  and Andr\'as Kov\'acs plays an important role in particular in
  connection with the open questions triggered by this paper.

\bibliographystyle{alpha}
\bibliography{int-lics}

\end{document}